\documentclass[12pt]{article}
\usepackage{tikz}
\usepackage[colorinlistoftodos,textsize=footnotesize]{todonotes}
\usetikzlibrary{matrix,arrows,shapes}
\usepackage{pgf}
\usepackage{cancel}
\usepackage{url}
\usepackage{longtable}
\usepackage{enumerate}
\usepackage{amsthm}
\usepackage{amssymb}
\usepackage{amsmath}
\usepackage{amscd}
\usepackage{eucal}
\usepackage{mathrsfs}
\usepackage{upgreek}
\usepackage{mathtools}
\usepackage{dsfont}
\usepackage{yfonts}
\usepackage[T1]{fontenc}
\usepackage{bbm}
\usepackage{geometry}
\usepackage{graphbox}
\usepackage{array}
\usepackage[retainorgcmds]{IEEEtrantools}
\usepackage{booktabs}
\usepackage{enumerate}
\usepackage{feynmf}
\usetikzlibrary{trees}
\usetikzlibrary{decorations.pathmorphing}
\usetikzlibrary{decorations.markings}
\tikzset{
	photon/.style={decorate, decoration={snake}, draw=red},
	electron/.style={draw=blue, postaction={decorate},
		decoration={markings,mark=at position .55 with {\arrow[draw=blue]{>}}}},
	gluon/.style={decorate, draw=magenta,
		decoration={coil,amplitude=4pt, segment length=5pt}},
	sderiv/.style={postaction={decorate},
		decoration={markings,mark=at position .3 with {\arrow{>}}}},
	tderiv/.style={postaction={decorate},
		decoration={markings,mark=at position .7 with {\arrow{<}}}},
	stderiv/.style={postaction={decorate},
		decoration={markings,mark=at position .7 with {\arrow{<}},mark=at position .3 with {\arrow{>}}}}
}
\usepackage{scalerel}
\usepackage{stackengine}

\definecolor{see}{RGB}{67,75,179}
\definecolor{darksee}{RGB}{42,44,148}
\definecolor{honey}{RGB}{232,180,129}
\definecolor{lighthoney}{RGB}{255,254,220}
\definecolor{citecol}{rgb}{0.5,0,0} 
\definecolor{blue1}{RGB}{130,150,209}
\usepackage{graphicx,stackengine,scalerel}
\def\tang{\ThisStyle{\abovebaseline[0pt]{\scalebox{-1}{$\SavedStyle\perp$}}}}
\definecolor{blue1}{RGB}{130,150,209}
\definecolor{blue2}{RGB}{42,60,122}
\usepackage{soul}
\usepackage[colorinlistoftodos,textsize=footnotesize]{todonotes}

\setstcolor{red}
\definecolor{see}{RGB}{67,75,179}
\usepackage{hyperref}
\hypersetup
{
	bookmarks=true,
	bookmarksopen=true,
	pdfmenubar=true, 
	pdfhighlight=/O, 
	colorlinks=true, 
	pdfpagemode=None, 
	pdfpagelayout=SinglePage, 
	pdffitwindow=true, 
	linkcolor=see, 
	citecolor=red!75!darkgray, 
	urlcolor=see 
}

\newcommand{\fA}{\mathfrak{A}}

\newcommand{\euR}{\mathscr{R}}  
\newcommand{\euS}{\mathscr{S}}

\newcommand{\euL}{\mathscr{L}}

\newcommand{\Ccal}{\mathcal{C}}
\newcommand{\Dcal}{\mathcal{D}}
\newcommand{\Ecal}{\mathcal{E}} 

\newcommand{\Fcal}{\mathcal{F}}

\newcommand{\Ncal}{\mathcal{N}}

\newcommand{\Ocal}{\mathcal{O}}
\newcommand{\Scal}{\mathcal{S}}
\newcommand{\Pcal}{\mathcal{P}}

\newcommand{\Tcal}{\mathcal{T}}
\newcommand{\Vcal}{\mathcal{V}}

\newcommand{\Xcal}{\mathcal{X}}
\newcommand{\Zcal}{\mathcal{Z}}
\newcommand{\Ycal}{\mathcal{Y}}

\newcommand{\Ci}{\mathcal{C}^\infty} 
\newcommand{\WF}{\mathrm{WF}}         
\newcommand{\id}{\mathrm{id}}               

\newcommand{\loc}{\mathrm{loc}}

\newcommand{\mc}{{\mu\mathrm{c}}}
\newcommand{\ml}{\mathrm{ml}}

\newcommand{\NN}{\mathbb{N}}          
\newcommand{\RR}{\mathbb{R}}           
\newcommand{\CC}{\mathbb{C}}           
\newcommand{\KK}{\mathbb{K}}   
\newcommand{\Co}{{\sst \CC}}
\newcommand{\M}{\mathbb{M}} 	     

\newcommand{\la}{\lambda}

\newcommand{\ph}{\varphi}

\newcommand{\T}{\cdot_{{}^\Tcal}}

\newcommand{\TT}{\Tcal}

\newcommand{\sst}[1]{\scriptscriptstyle{#1}}  
\newcommand{\1}{1}                         
\newcommand{\be}{\begin{equation}}
\newcommand{\ee}{\end{equation}}

\DeclareMathOperator{\supp}{supp}      

\theoremstyle{plain}
\newtheorem{thm}{Theorem}[section]
\newtheorem{prop}[thm]{Proposition}
\newtheorem{obs}[thm]{Observation}

\newtheorem{lemma}[thm]{Lemma}

\theoremstyle{definition}
\newtheorem{df}[thm]{Definition}
\theoremstyle{remark}
\newtheorem{rem}[thm]{Remark}

\begin{document}		
		\title{Locality and causality in perturbative algebraic quantum field theory}
	\author{Kasia Rejzner}
	\date{\today}

	\maketitle
	
	\begin{abstract}
		In this paper we discuss how seemingly different notions of locality and causality in quantum field theory can be unified using a non-abelian generalization of the Hammerstein property (originally introduced as a weaker version of linearity). We also prove a generalization of the main theorem of renormalization, in which we do not require field independence.
	\end{abstract}

	\tableofcontents

\section{Introduction}
In this article we show how the different meanings of the term ``locality'' appearing in quantum field theory (QFT) can be described using appropriate versions of \textit{the Hammerstein property} (term used in \cite{Batt} to describe a weaker notion of linearity), when applied to groups and vectors spaces equipped with binary relations of a certain kind.
This is based on ideas proposed in \cite{BDF,CGPZ,CGPZ18}. The key new result of the paper is theorem \ref{thm:main}, which is a generalization of the main theorem of renormalization proven in \cite{BDF}, where we drop some of the axioms and state the theorem in a more general context, without reference to functionals on the space of fields.

 It was argued in \cite{BDF} that the notion of locality for smooth functionals on the space of field configurations is captured by the following property (Hammerstein property):
\be\label{add}
F(\ph_1+\ph_2+\ph_3)=F(\ph_1+\ph_2)+F(\ph_2+\ph_3)-F(\ph_2)\,,
\ee
if $\supp(\ph_1)\cap\supp(\ph_3)=\varnothing$, where $\ph_1,\ph_2,\ph_3\in \Ci(M,\RR)$, $F\in\Ci(\Ci(M,\RR),\CC)$, and $M$ is a Lorentzian manifold (here we typically assume $M$ to be globally hyperbolic) In \cite{BDF} this property was termed \textit{Additivity}, as it can be seen as generalization of linearity. Indeed, it is clearly satisfied for $F$ linear.

For $F$ with $F(0)=0$, property \eqref{add} implies that
\be\label{padd}
F(\ph_1+\ph_3)=F(\ph_1)+F(\ph_3)\,,
\ee
if $\supp(\ph_1)\cap\supp(\ph_3)=\varnothing$. This condition is called \textit{disjoint additivity} (in \cite{BDGR} referred to as \textit{partial additivity}). It has been shown in \cite{BDGR} (by a counter example) that this condition is strictly weaker than \eqref{add}.  However, for polynomial unit-preserving (i.e. $F(0)=0$) functionals \eqref{add} and \eqref{padd} are equivalent, as shown in \cite{BDF}.

In \cite{BDGR} it was shown, following the sketch of the proof given in an earlier version of \cite{BFR}, that \eqref{add} together with an additional regularity condition that we will recall in section~\ref{sec:smooth:func} is equivalent to saying that $F$ can be written as an integral of some smooth function on the jet space (this is \textit{locality} for functionals). In the final version of  \cite{BFR}, this condition was slightly weakened.

In a seemingly different context, the term \textit{locality} is used in AQFT to express the fact that local algebras $\fA(\Ocal_1)$ and $\fA(\Ocal_2)$, assigned to bounded regions $\Ocal_1,\Ocal_2\subset\M$ of Minkowski spacetime $\M$, commute:
\[
[\fA(\Ocal_1),\fA(\Ocal_2)]=\{0\}\,,
\]
if $\Ocal_1$ is spacelike to $\Ocal_2$. This property is also called \textit{Einstein causality}. We will show here that these two notions of locality can be brought together in perturbative AQFT (pAQFT), when applied to formal S-matrices.

We claim that the physical notion of \textit{locality} is well captured by what we refer to as the \textit{generalized Hammerstein property}.
\section{Locality and causality structures}
In \cite{CGPZ} the authors introduce an abstract notion of locality captured by the definition of the \textit{locality set}:

\begin{df}\begin{enumerate}[(a)]\label{df:localityset}
\item	A locality set is a pair $(X,\tang)$ where $X$ is a set and $\tang$ is a symmetric binary relation on $X$.
For $x_1, x_2 \in X$, denote $x_1\tang x_2$ if $(x_1,x_2) \in\tang$. 
\item For any subset $U \subset X$, let
\[U^{\tang} := \{x \in X |\, x\tang y\,\ \textrm{for all}\ y\in U\}\]
	denote the polar subset of $U$.
\item	For integers $k \geq 2$, denote
\[	
X^{\tang k}\equiv X\times_{\tang}\dots \times_{\tang} X:= \{(x_1,\dots x_k) \in X^k |\  x_i\tang x_j\,,\ 1 \leq i \neq j \leq k\}\,.
\]
\item We call two subsets $A$ and $B$ of a locality set $(X,\tang)$, \textit{independent}, if $A \times B \subset \tang$.
\end{enumerate}
\end{df}

An example of such relation is disjointness of sets. In particular, the primary example we are interested in is the power set $\Pcal(M)$, of a manifold $M$ equipped with the \textit{disjointness of sets} relation $\tang$.

 If $M$ is an oriented, time-oriented Lorentzian spacetime ($M$ is equipped with a smooth Lorentzian metric $g$, a nondegenerate volume form and a smooth vector field $u$ on $M$ such that for every $p\in M$, $g(u,u)>0$) the notion of disjointness of sets can be refined with the use of the causal structure.  A curve $\gamma:\RR\supset I\rightarrow M$ with a tangent vector $\dot{\gamma}$ is:
\begin{itemize}
	\item spacelike if $g(\dot{\gamma},\dot{\gamma})<0$,
	\item  timelike if $g(\dot{\gamma},\dot{\gamma})>0$,
	\item  lightlike if $g(\dot{\gamma},\dot{\gamma})=0$,
	\item  causal if $g(\dot{\gamma},\dot{\gamma})\geq0$.
\end{itemize}
A causal/lightlike/timelike curve $\gamma$ is called \textit{future-pointing} if $g(u,\dot{\gamma})>0$, where $u$ is the time orientation.
\begin{df}
 	Let $x\in M$, where $M$ is a Lorentzian oriented, time-oriented spacetime. Let $J^\pm(x)$, the future/past of $x$, be the set of all points $y\in M$ such that there exists a future/past pointing causal curve from $x$ to $y$. 
 \end{df}
\begin{df}
	Let $\Ocal_1,\Ocal_2\subset M$, where $M$ is a Lorentzian manifold. One says that $\Ocal_1$ ``is not later than'' $\Ocal_2$, i.e. $\Ocal_1\preceq \Ocal_2$ if and only if $\Ocal_1\cap J^+(\Ocal_2)=\varnothing$.
\end{df}
This relation is crucial for formulation of the causal factorization property of formal S-matrices (we will come to this in Section~\ref{FS}). Clearly, this is not a symmetric relation, but it can be symmetrized to obtain:
\begin{obs}
	Symmetrization of  $\preceq$ is the relation of being \textit{spacelike}, $\bigtimes$, i.e. 
	\[
	\Ocal_1\bigtimes\Ocal_2 \Leftrightarrow \Ocal_1\preceq\Ocal_2 \wedge \Ocal_2\preceq\Ocal_1\,.
	\]
\end{obs}
There are two features of $\simeq$ that we want to single out as important: each subset $\Ocal\subset M$ has non-zero intersection with both the past and the future of itself, so $\Ocal\sim\preceq \Ocal$. There exist subsets for which the relation is non symmetric, i.e. $\Ocal_1\preceq \Ocal_2$, but  $\Ocal_2\npreceq \Ocal_1$ (existence of ``preferred direction'').

Motivated by the above example we define:
\begin{df}
	\begin{enumerate}[(a)]
		\item 	A causality set is $(X,\dashv)$, where $X$ is a set and $\dashv$ a relation such that for all $x\in X$ $x \sim\dashv x$ (the negation of $\dashv$ is reflexive) and $\dashv$ is not symmetric (i.e. there exist $x,y\in X$ such that $x\dashv y$ and $y\sim\dashv x$).
		\item For any subset $U \subset X$, let
		\begin{align*}
		{}^{\dashv}U &:= \{x \in X | x \dashv y, \textrm{ for all }y\in U\}\,,\\
		U^{\dashv}&:= \{x \in X | y \dashv x, \textrm{ for all }y\in U\}\,.
		\end{align*}
		
	\end{enumerate}
\end{df}
\begin{obs}
	Given a causality set $(X,\dashv)$ we obtain a locality set by symmetrizing $\dashv$.
\end{obs}

In \cite{CGPZ}, the authors introduce the notion of a \textit{locality group}, which is a set $(G,\tang)$ together with a product law $m_G$ defined on $\tang$, for which the product is compatible with the locality relation on $G$ and the other group properties (associativity, existence of the unit, inverse) hold in a restricted sense, again compatible with $\tang$. This structure is an analogue of a \textit{partial algebra} \cite{KS75}. In fact, a partial algebra of functionals equipped with the time-ordered product discussed in \cite{KS75} is a locality algebra in the sense of \cite{CGPZ}. We will discuss this in more detail in section \ref{sec:conc:mod}.

However, in our context, we will need to equip a locality set $(G,\tang)$ with a full group structure (without restricting to $\tang$), compatible with $\tang$. To make a distinction to \cite{CGPZ} we will call this a \textit{group with locality}.

\begin{df} \label{defn:lgr}
	A \textit{group with locality}  $(G,\tang,m_G,\1_G)$ is a group $(G,m_G,\1_G)$ and a locality set $(G,\tang)$, for which the product law $m_G$ is compatible with $\tang$, i.e.
		\begin{equation}\text{ for all } U\subseteq G, \quad  m_G((U^{\tang}\times U^{\tang})\cap{\tang})\subset U^{\tang}
		\end{equation}
		and such that $\{\1_G\}^{\tang}=G$.
\end{df}
Similarly one can also define \textit{vector spaces with locality}. The obvious analog of definition~\ref{defn:lgr} coincides in fact with what \cite{CGPZ} call \textit{locality vector spaces}.
\begin{df}\label{defn:lvs}
		A {\it vector space with locality (locality vector space)} is a vector space $V$ equipped with a locality relation $\tang$ which is compatible with the linear structure on $V$ in the sense that, for any  subset $X$ of $V$, $X^{\tang}$ is a linear subspace of $V$.
\end{df}

Let $M$ be a Lorentzian manifold.
The set $\Ecal(M)\doteq \Ci(M,\RR)$ together with the disjointness of supports relation $\tang$ and with addition $(f,g)\mapsto f+g$ forms a  group with locality $(\Ecal(M),\tang,+,0)$. Moreover, equipped with multiplication by scalars in $\RR$, it is a locality vector space.

One can formulate the notion of disjoint additivity in terms of morphisms of groups with locality. We recall after \cite{CGPZ}:
\begin{df}
	A \emph{locality map} from a locality set $\left(X,{\tang}_X\right)$ to a locality set $ (Y, {\tang}_Y)$ is a map $\phi:X\to Y$ such that $(\phi\times \phi)({\tang}_X)\subseteq {\tang}_Y$. 
	\label{defn:localmap}
\end{df}
With a slight modification of the definition of \emph{locality morphisms} from \cite{CGPZ}, we introduce
\begin{df} Let $\left(X,{\tang}_X,\cdot_X,\1_X  \right) $ and $\left(Y,{\tang}_Y,\cdot_Y,\1_Y  \right) $ be groups with locality. A map $\phi:X\longrightarrow Y$ is called a \textit{morphism of groups with locality}, if it
	\begin{enumerate}
		\item is a locality map;
		\item is \emph{locality multiplicative}:
		for $({ a},{ b})\in {\tang}_X$ we have
		$\phi({ a}\cdot_X{ b})= \phi({ a})\cdot_Y\phi({ b})$,
		\item preserves the unit  $\phi(\1_X)=\1_Y$.
	\end{enumerate}
\end{df} 

\begin{df}
Let $\left(X,{\tang}_X,+_X,0_X,\cdot_X,\KK \right) $ and $\left(Y,{\tang}_Y,+_Y,0_Y,\cdot_Y,\KK \right) $  be vector spaces with locality over a field $\KK$. A map $\phi:X\longrightarrow Y$ is called a \textit{morphism of vector spaces with locality}, if it is a morphism of groups with locality between $\left(X,{\tang}_X,+_X,0_X\right) $ and $\left(Y,{\tang}_Y,+_Y,0_Y\right) $.
\end{df}

We equip $\RR$ with the trivial locality relation $\tang_\RR=\RR\times \RR$ and the locality group structure $(\RR,\tang_{\RR},+,0)$. 
\begin{obs} A functional $F:\Ecal(M)\rightarrow \RR$ is disjoint additive \eqref{padd} and satisfies  $F(0)=0$, if and only if $F$ is a morphism with locality between the groups $(\Ecal(M),\tang,+,0)$ and $(\RR,\tang_\RR,+,0)$.
\end{obs}
Similar to a group with locality, we can also define a \textit{group with causality}. This notion will be crucial in section \ref{FS}.

\begin{df} \label{defn:cgr}
	A \textit{group with causality}  $(G,\dashv,m_G,\1_G)$ is a group $(G,m_G,\1_G)$ and a causality set $(G,\dashv)$, for which the product law $m_G$ is compatible with $\dashv$, i.e.
\begin{itemize}
	\item if $x_1\dashv y$ and $x_2\dashv y$, then $m_G(x_1,x_2)\dashv y$.
	\item if $y\dashv x_1$ and $y\dashv x_2$, then $y\dashv m_G(x_1,x_2)$.
	\item $\{\1_G\}^{\dashv}=G$ and ${}^{\dashv}\{\1_G\}=G$.
\end{itemize}
\end{df}
\begin{df} Let $\left(X,{\dashv}_X,\cdot_X,\1_X  \right) $ and $\left(Y,{\dashv}_Y,\cdot_Y,\1_Y  \right) $ be groups with causality. A map $\phi:X\longrightarrow Y$ is called a \textit{morphism of groups with causality}, if it
	\begin{enumerate}
		\item is a causality map (i.e. if $a\dashv_X b$ then  $\phi(a)\dashv_Y \phi(b)$);
		\item is \emph{causality multiplicative}:
		for $b\dashv_X a$ we have
		$\phi({ a}\cdot_X{ b})= \phi({ a})\cdot_Y\phi({ b})$,
		\item preserves the unit  $\phi(\1_X)=\1_Y$.
	\end{enumerate}
\end{df}
\begin{rem}
	The order of factors in the definition of being \textit{causality multiplicative} is chosen to agree with the convention used in pAQFT.
\end{rem}
\begin{df}\label{defn:cvs}
	A {\it vector space with causality} is a vector space $V$ equipped with a causality relation $\tang$ which is compatible with the linear structure on $V$. A morphism is called a \textit{morphism of vector spaces with causality} if it is a group with causality morphism for underlying additive groups.
\end{df}

A notion stronger than being locality or causality multiplicative is provided by the \textit{Hammerstein property}. Here we state it for groups, but one can also introduce it in other contexts, e.g. vector spaces. First we need one more definition.
\begin{df}\label{def:relative}
	Let $\left(X,{\cdot}_X,\1_X  \right)$, $\left(Y,{\cdot}_Y,\1_Y  \right)$ be groups and let $\phi: X\rightarrow Y$. Given $a\in X$, one defines the map $\phi_a:X\rightarrow Y$ by
	\[
	\phi_a(b)\doteq \phi(a)^{-1}\cdot_Y\phi(a\cdot_X b)\,,
	\]
	where $b\in X$.
\end{df}
\begin{rem}
	For non-abelian groups one can consider variants of this definition, where the order of terms is interchanged. Our choice of convention is motivated by the one used in QFT for the definition of the \emph{relative $S$-matrix} (see \cite{BDF} and a brief discussion at the end of Section~\ref{sec:conc:mod}) 
\end{rem}

\begin{df}[Generalized Hammerstein property]\label{def:genHamm} Consider a group  $\left(X,{\cdot}_X,\1_X  \right)$ with causality/locality specified by the relation ${\dashv}_X$/${\tang}_X$, respectively. Let $\left(Y,{\cdot}_Y,\1_Y  \right)$ be another group. A map $\phi:X\longrightarrow Y$ satisfies the \textit{generalized Hammerstein property}, if $\phi_a:X\longrightarrow Y$ (as given in Definition~\ref{def:relative}) is causality/locality multiplicative for all $a\in X$.
\end{df}
\begin{rem}
	For commutative  $\left(X,+,0,\tang_X  \right)$, the generalized Hammerstein property implies in particular that for $b\tang_X c$ we have the commutativity of the images:
	\[
	\phi_a(b+c)=\phi_a(b)\cdot_Y\phi_a(c)=\phi_a(c)\cdot_Y\phi_a(b)
	\]
	for all $a\in X$.
\end{rem}

\begin{rem}
For commutative  $\left(X,+,0,\dashv_X\right)$, Definition~\ref{def:genHamm} is equivalent to the condition that for $c\dashv_X b$ we have
	\[\phi(a+b+c)= \phi({ a+ b})\cdot_Y\phi(a)^{-1}\cdot_Y\phi({ a + c})\,,
	\]
which will be crucial in the definition of local S-matrices in section \ref{FS}. 
\end{rem}

We believe that the Hammerstein property in some sense singles out structures that we intuitively describe as \textit{local}. In this note we show that it features in the definition of:
\begin{itemize}
	\item Local functionals,
		\item Local Haag-Kastler nets,
	\item Local $S$-matrices,
	\item Local renormalization maps.
\end{itemize}
\section{Haag-Kastler axioms}
In the Haag-Kastler axiomatic approach to quantum field theory \cite{HK}, one describes a model by assigning algebras $\fA(\Ocal)$ (originally, $C^*$ or von Neumann algebras) to bounded regions $\Ocal$ in Minkowski spacetime $\M$. Hence this framework is often called algebraic quantum field theory (AQFT). More generally, in locally covariant QFT \cite{BFV,HW,HW01}, one can work with causally convex relatively compact subsets of a globally hyperbolic spacetime $M$.

This assignment of algebras to regions (the net of algebras) has to fulfill several physical requirements, among them:
\begin{enumerate}[{\bf HK1}]
	\item {\bf Isotony:} if $\Ocal_1\subset\Ocal_2$, then $\fA(\Ocal_1)\subset\fA(\Ocal_2)$.\label{HK:iso}
	\item {\bf Locality (Einstein causality):} if $\Ocal_1\bigtimes\Ocal_2$, then $[\fA(\Ocal_1),\fA(\Ocal_2)]_{\fA(\Ocal)}=\{0\}$, where $\Ocal$ is any  causally convex relatively compact  region containing both $\Ocal_1$ and $\Ocal_2$.\label{HK:loc}
	\item {\bf Time-slice axiom:}  if $\Ncal\subset \Ocal$ is a neighborhood of a Cauchy surface of $\Ocal$, then $\fA(\Ncal)=\fA(\Ocal)$.\label{HK:ts}
\end{enumerate}
In perturbative AQFT (see e.g. \cite{BDF,DF,BF0,Book}), algebras $\fA(\Ocal)$ are considered formal power series with coefficients in toplogical star algebras. In this approach, one can construct the local algebras $\fA(\Ocal)$ using the concept of \textit{formal $S$-matrices}.
\section{Formal $S$-matrices}\label{FS}
Following \cite{Vienna,BF19}, we review the construction of the net of algebras satisfying {\bf HK\ref{HK:iso}-HK\ref{HK:ts}}, using formal $S$-matrices.
\begin{df}[Generalized local $S$-matrix]\label{Smatrix}
	Let $(G,\dashv,+,0)$ be a 
	 group with causality and $\fA$ a unital topological *-algebra, with $U(\fA)\subset\fA$ denoting its group of unitary elements. A map $\Scal:G\rightarrow U(\fA)$ is a \textit{generalized local $S$-matrix} on $(G,\dashv,+,0)$  (or labeled by the elements of $G$) if it fulfills the following axioms:
	\begin{enumerate}[{\bf S1}]
		\item {\bf Identity preserving}:\label{S:start} $\Scal(0)=\1$.
		\item {\bf Locality}:\label{S:loc} $\Scal$ satisfies the Hammerstein property with non-commutative target (Def.~\ref{def:genHamm}), i.e. 
		$f_1\dashv f_2$ implies that
		\[
		\Scal(f_1+f+f_2)=	\Scal(f_2+f)\Scal(f)^{-1}	\Scal(f+f_1)\,,
		\]
		where $f_1, f, f_2\in G$.
	\end{enumerate}
Let $\euS_G$ denote the space of all generalized local $S$-matrices for the given group $G$ with causality.
\end{df} 
There are some further physically motivated axioms that one can impose, related to the dynamics. We will discuss this in Section~\ref{sec:local:S:matrices}, following the approach of \cite{BF19}.

\begin{rem}
Given a group with causality $(G,\dashv,+,0)$ it is easy to obtain a generalized local S-matrix by setting $\fA$ to be $\fA_G$, the group algebra over $\CC$ of the free group generated by elements $\Scal(f)$ (these are now formal generators), $f\in G$, modulo relations {\bf S\ref{S:loc}
} and {\bf S\ref{S:start}} (see e.g. \cite{BF19}).
 \end{rem}

Note that Definition~\ref{Smatrix} implies in particular that $\Scal$ is \textit{unit preserving} and \textit{causality multiplicative}, since for $f_2\dashv f_1$ we have
\[
\Scal(f_1+f_2)=	\Scal(f_1)\Scal(f_2)\,.
\]
By symmetrizing $\dashv$, we obtain from  $(G,\dashv,+,0)$ a group $(G,\tang,+,0)$ with locality relation $\tang$ defined by:
$f_1\tang f_2$ if both $f_1\dashv f_2$ and $f_2\dashv f_1$. It follows that a generalized local S-matrix $\Scal:(G,\tang,+,0)\rightarrow U(\fA)$ is \textit{locality multiplicative}, i.e. if $f_1\tang f_2$ then
	\[
\Scal(f_1+f_2)=	\Scal(f_1)\Scal(f_2)=\Scal(f_2)\Scal(f_1)\,,
\]
hence
\[
[\Scal(f_1),\Scal(f_2)]=0\,.
\]

Consider $(\Dcal(M),\preceq,+,0)$, the space of test functions $\Dcal(M)\doteq \Ci_c(M,\RR)$ with the additive group structure and the causality relation $\preceq$  of being \textit{not in the future of}. Clearly, if $\Scal$ is a formal $S$-matrix on $(\Dcal(M),\preceq,+,0)$ and we define $\fA(\Ocal)$ as the algebra generated by $\Scal(f)$, where $\supp(f)\subset\Ocal$, then the net $\Ocal\mapsto\fA(\Ocal)$ satisfies {\bf HK\ref{HK:loc}}.

This is how one can connect the abstract notion of locality expressed in terms of the Hammerstein property to the locality property {\bf HK\ref{HK:loc}} in terms of Haag-Kastler axioms.
\subsection{Renormalization group}
\subsubsection{Abstract definition}
\begin{df}[Renormalization group]\label{Rgroup}
	Let $(G,\dashv,+,0)$ be a group with causality. The renormalization group $\euR_G$ for $G$ is the group of maps $\Zcal:G\rightarrow G$, which are:
	\begin{enumerate}[{\bf Z1}]
		\item {\bf Identity preserving}: $\Zcal(0)=0$.\label{Z:ident}
		\item {\bf Causality preserving}:\label{Z:caus:pres} $\Zcal_{f}$ (as in Def.~\ref{def:relative}) satisfies the following:
		\[
		f_1\dashv f_2\ \textrm{implies}\ \Zcal_f(f_1) \dashv f_2\ \textrm{and}\ f_1\dashv\Zcal_f(f_2)\,.
		\]
		\item {\bf Local}:\label{Z:loc} $\Zcal$ 
		satisfies the Hammerstein property (Def.~\ref{def:genHamm}), i.e. $f_1\dashv f_2$ implies that
		\[
		\Zcal(f_1+f+f_2)=	\Zcal(f_1+f)-\Zcal(f)+	\Zcal(f_2+f)\,,
		\]
		where $f_1, f_2, f\in G$.
	\end{enumerate}
\end{df}
\begin{rem}
	Note that {\bf Z\ref{Z:caus:pres}} implies that $\Zcal$ is a causality map.
\end{rem}
Clearly, if $\Zcal\in\euR_G$, and $\Scal\in\euS_G$, then $\Scal\circ\Zcal$ is also an S-matrix, hence:
\begin{prop}
	$\euS_G$ is an $\euR_G$-module.
\end{prop}
\begin{proof}
The only non-trivial check is the Locality~\textbf{S\ref{S:loc}}. For $f_2\dashv f_1$ we have
\begin{multline*}
\widetilde{\Scal}(f_1+f+f_2)=\Scal\circ \Zcal (f_1+f+f_2)
=\Scal(\Zcal_f(f_1)+\Zcal_f(f_2)+\Zcal(f))\\
= 
\Scal(\Zcal_f(f_1)+\Zcal(f))\Scal(\Zcal(f))^{-1}\Scal(\Zcal_f(f_2)+\Zcal(f))
=\widetilde{\Scal}(f_1+f)\widetilde{\Scal}(f)^{-1}\widetilde{\Scal}(f_2+f)\,.
\end{multline*}
where in the second step we used \textbf{Z\ref{Z:loc}} while in the third step we used \textbf{Z\ref{Z:caus:pres}} together with \textbf{S\ref{S:loc}}.
\end{proof}

One can now ask the question whether for given $\Scal,\widetilde{\Scal}\in \euS_G$, there exists a $\Zcal\in \euR_G$ such that $\widetilde{\Scal}=\Scal\circ \Zcal$ (main theorem of renormalization). This is 
 more tricky to show and has been proven only under some additional assumptions, e.g. in the perturbative setting of \cite{BDF}. It would be interesting to investigate this problem for a finite dimensional Lie group $G$ with additional requirement that elements of $\euS_G$ and $\euR_G$ are analytic maps. In the next section we present a new variant of the perturbative setting of of \cite{BDF}, where we do not require that $G$ is a set of functionals on the space of fields. We then proceed to prove the main theorem of renormalization in this more general setting.

\subsubsection{Perturbative renormalization group}\label{prg}
We start with an algebra $\fA[[\lambda]]$ (where $\lambda$ is interpreted as the coupling constant) with a causality relation $\dashv$, such that the underlying vector space of $\fA$ equipped with $\dashv$ is a vector space with causality. Given the relation $\dashv$ on $\fA$, we extend it to $\fA[[\lambda]]$ by setting 
\[
\sum_{n=0}^{\infty} \lambda^n f_n\dashv \sum_{n=0}^{\infty} \lambda^n g_n\,,\quad \textrm{iff}\quad f_n\dashv g_m\ \forall n,m\in\NN_0.
\]
Let $\Vcal$ be a subspace of $\fA$ with the property:
\begin{enumerate}[{\bf L1}]
\item \label{L1} $\forall g\in \Vcal$, given $f_1,f_2\in \fA$ such that $f_1\dashv f_2$ and $N\in \NN$, $N>3$, there exist $g_1,\dots, g_N\in \Vcal$ (these do not need to be distinct) such that $g=\sum_{i=1}^{N} g_i$, the decomposition preserves $\dashv$ (i.e. if $g\dashv f$ then $g_i\dashv f$, for all $i$ and  if $h\dashv g$ then $h\dashv g_i$ for all $i$) and the following hold: 
\begin{align}
&g_i\dashv f_2\ \textrm{for}\ i=1,\dots,N-1\,,\label{rel1}\\ 
&f_1\dashv g_i\ \textrm{for}\ i=2,\dots,N\label{rel2}\\
&g_i\dashv g_j\ \textrm{for}\ 2\leq i+1<j\leq N\,.\label{rel3}
\end{align}
\end{enumerate}
\begin{rem}
	To see that elements fulfilling {\bf L\ref{L1}} indeed form a subgroup, note that if $g,h\in \Vcal$, then we can decompose $g=\sum_{i=1}^{N} g_i$ and decompose $h$ using all the relations \eqref{rel1}-\eqref{rel3} for $g_i$'s. By combining the terms appropriately, we obtain a decomposition of  $g+h$ satisfying  \eqref{rel1}-\eqref{rel3}. Observe that $(\{0\},+,0,\dashv)$ has {\bf L\ref{L1}}, so $\Vcal$ always exists.
\end{rem}
\begin{rem}
	Note that {\bf L\ref{L1}} implies that $\forall g\in \Vcal$, given $f_1,f_2\in \fA$ such that $f_1\dashv f_2$ and $N\in \NN$, $N>3$, there exist $g_1,\dots, g_N\in\fA$ such that $g=\sum_{i=1}^{N} g_i$ and for all subsets $I\subsetneq \{1,\dots,N\}$, there exist $I_1,I_2\subset \{1,\dots,N\}$ such that $I=I_1\dot\cup I_2$ and all $f_1,\{g_i,i\in I_1\}$ are $\dashv$ with respect to all $f_2,\{g_i,i\in I_2\}$. This is analogous to a property proven in \cite[Lemma 3.2]{BDF} for local functionals and we will use it in the proof of Lemma~\ref{key:lemma} later on. 
\end{rem}

In definition \ref{Rgroup}, we replace the group with the vector space $\lambda \Vcal[[\lambda]]$. We write the generalized local S-matrices $\Scal: \lambda \Vcal[[\lambda]]\rightarrow \fA[[\lambda]]$ as
\[
\Scal(\lambda f)=\1+\sum_{n=1}^{\infty} \frac{\lambda^n}{n!}\Tcal_n(f^{\otimes n})\,,
\]
where $f\in \Vcal$ and call the multilinear maps $\Tcal_n:\Vcal^{\otimes n}\rightarrow \fA$, \textit{the $n$-fold time-ordered products} (they extend to maps on $\lambda \Vcal[[\lambda]]$ in the obvious way). 

 We require in addition to the previous axioms for local S-matrices that 

\begin{enumerate}[{\bf S1}]
	\addtocounter{enumi}{2}
	\item \label{S:start2} $\TT_1=\id$.
\item {\bf Causality preserving}:\label{S:caus:pres} $\Scal$ satisfies the following:
\[
f_1\dashv f_2\ \textrm{implies}\ \Scal(f_1) \dashv f_2\ \textrm{and}\ f_1\dashv\Scal(f_2)\,.
\]
\end{enumerate}
From renormalization group elements, we require in addition that:
\begin{enumerate}[{\bf Z1}]
	\addtocounter{enumi}{3}
	\item \label{Z:order} $\Zcal=\id + \Ocal(\lambda)$.
\end{enumerate}
They are then given in terms of formal power series, so that
\[
\Zcal(\lambda f)=f+\sum_{n=2}^{\infty} \frac{\lambda^n}{n!}\Zcal_n(f^{\otimes n})\,,
\]
where $f\in \Vcal$.

The following theorem generalizes the main theorem of renormalization that has been proven in \cite{BDF} in the context of QFT:
\begin{thm}\label{thm:main}
	Let $\fA[[\lambda]]$ and $\lambda \Vcal[[\lambda]]$ be as specified above. Given two generalized S-matrices $\Scal, \widetilde{\Scal}\in\euS$ (in the sense of definition~\ref{Smatrix} for vector spaces instead of groups, with additional properties  {\bf S\ref{S:start2}} and {\bf S\ref{S:caus:pres}}), there exists an element of the renormalization group $\Zcal\in \euR$ (satisfying properties {\bf Z\ref{Z:ident}}-{\bf Z\ref{Z:order}}), such that
	\[
	\widetilde{\Scal}=\Scal\circ \Zcal\,.
	\]
	Also the converse holds.
\end{thm}
The following lemmas will be useful in the proof of the theorem:
\begin{lemma}\label{key:lemma}
	Let $\fA[[\lambda]]$ and $\lambda \Vcal[[\lambda]]$ be as specified above. If a map $Z:\lambda \Vcal[[\lambda]]\rightarrow \fA[[\lambda]]$ satisfies {\bf Z\ref{Z:caus:pres}} and {\bf Z\ref{Z:loc}} for $f=0$, it satisfies  {\bf Z\ref{Z:caus:pres}} and {\bf Z\ref{Z:loc}} for general $f\in \lambda \Vcal[[\lambda]]$.
\end{lemma}
\begin{proof}
	We use the property {\bf L\ref{L1}} and the fact that $Z(f+\mu g)$ is determined (in the sense of formal power series) by the derivatives at $\mu=0$. Following the argument used in Appendix B of \cite{BDF}, we prove that the $n$-th derivative with respect to $\mu$ of both sides of
	\[
	Z(f_1+\mu g+f_2)=Z(f_1+\mu g)-Z(\mu g)+Z(\mu g+f_2)
	\]
	vanishes for all $n>1$. Take $N>n$. It follows from {\bf L\ref{L1}} that $g$ can be written as
	\[
	g=\sum_{i=1}^{N}g_i
	\]
	such that for all subsets $I\subsetneq \{1,\dots,N\}$, there exist $I_1,I_2\subset \{1,\dots,N\}$ such that $I=I_1\dot\cup I_2$ and all $f_1,\{g_i,i\in I_1\}$ are $\dashv$ with respect to all $f_2,\{g_i,i\in I_2\}$. Let $\alpha=(\alpha_1,\dots,\alpha_N)$ with $|\alpha|\leq n$ and $\mu=(\mu_1,\dots,\mu_N)$. Consider
	\[\label{eq:derivative}
	\partial_\mu^\alpha \big(Z(f_1+g(\mu)+f_2)-Z(f_1+g(\mu))+Z( g(\mu))-Z(g(\mu)+f_2)\big)\big|_{\mu=0}\,,
	\]
	where $g(\mu)=\sum_{i=1}^{N}\mu_i g_i$.
	Let $I=\{i\in\{1,\dots,N\}|\alpha_i\neq 0\}$. Note that $I\subsetneq \{1,\dots,N\}$, so we can choose a decomposition $I=I_1\dot\cup I_2$ as described above. Let $g_{I_1}(\mu)\doteq \sum_{i\in I_1} \mu_i g_i$ and $g_{I_2}(\mu)\doteq \sum_{i\in I_2} \mu_i g_i$. The derivative in \eqref{eq:derivative} can be now written as
	\begin{multline*}
	\partial_\mu^\alpha \big(Z(f_1+g_{I_1}(\mu)+g_{I_2}(\mu)+f_2)
	-Z(f_1+g_{I_1}(\mu)+g_{I_2}(\mu))+Z(g(\mu))
	-Z(g_{I_1}(\mu)+g_{I_2}(\mu)+f_2)\big)\big|_{\mu=0}\\
	= \partial_\mu^\alpha \big(Z(f_1+g_{I_1}(\mu))+Z(g_{I_2}(\mu)+f_2)
	-Z(f_1+g_{I_1}(\mu))-Z(g_{I_2}(\mu))+Z(g_{I_1}(\mu))\\
	+Z(g_{I_2}(\mu))-Z(g_{I_1}(\mu))-Z(g_{I_2}(\mu)+f_2)\big)\big|_{\mu=0}=0\,.
	\end{multline*}
	where in the last step we used {\bf Z\ref{Z:loc}} for $f=0$.
	
	We follow a similar strategy to prove {\bf Z\ref{Z:caus:pres}}. With the notation as above, we want to show that for $f_1\dashv f_2$ and for all $n$, $|\alpha|\leq n$, we have
	\[
	\partial_\mu^\alpha (Z(f_1+g(\mu))-Z(g(\mu)))\dashv f_2\,,\qquad
	f_1\dashv \partial_\mu^\alpha (Z(f_2+g(\mu))-Z(g(\mu)))\,.
	\]
	As before, we can replace $g(\mu)$ with $g_{I_1}(\mu)+g_{I_2}(\mu)$ and use  {\bf Z\ref{Z:loc}} for $f=0$ to write the above condition as
	\[
	\partial_\mu^\alpha (Z(f_1+g_{I_1}(\mu))-Z(g_{I_1}(\mu)))\dashv f_2\,,\qquad	f_1\dashv \partial_\mu^\alpha (Z(f_2+g_{I_2}(\mu))-Z(g_{I_2}(\mu)))\,,
	\]
	and we see that it is indeed satisfied due to {\bf Z\ref{Z:caus:pres}} for $f=0$.
\end{proof}
\begin{lemma}\label{key:lemma2}
		Let $\fA[[\lambda]]$ and $\lambda \Vcal[[\lambda]]$ be as specified above. If a map $Z:\lambda \Vcal[[\lambda]]\rightarrow \fA[[\lambda]]$ satisfies {\bf Z\ref{Z:caus:pres}} and {\bf Z\ref{Z:loc}}, then $Z(f)\in\lambda \Vcal[[\lambda]]$ for all $f\in\lambda \Vcal[[\lambda]]$.
\end{lemma}
\begin{proof}
	We have to show that, given $f_1\dashv f_2$, $Z(g)$ can be decomposed as $Z(g)=\sum_{i=1}^N z_i$ for any $N>3$, so that the properties stated in {\bf L\ref{L1}} are satisfied. We do it by construction. Take a decomposition $g=\sum_{i=1}^N g_i$ that satisfies the properties required by {\bf L\ref{L1}}. Using {\bf Z\ref{Z:loc}} we can decompose
	\[
	Z(\sum_{i=1}^N g_i)=\sum_{i=1}^{N} z_i\,,
	\]
	where $z_i=Z_{g_{i+1}}(g_i)$ for $i=1,\dots, N-1$ and $z_N=Z(g_N)$. It follows now from {\bf Z\ref{Z:caus:pres}} that this decomposition satisfies the requirements \eqref{rel1}-\eqref{rel3}. Moreover, each of the terms $z_i$ can be further decomposed, since each $g_i$ can be decomposed in a way compatible with $\dashv$.
\end{proof}
\begin{proof}(\textit{of the theorem})
	We follow directly the inductive proof of \cite{BDF}.
	Given $\widetilde{\Scal}$ and $\Scal$, we want to construct $\Zcal$ as a formal power series:
	\[
	\Zcal(f)=f+\sum_{n=2}^{\infty} \frac{\lambda^n}{n!}\Zcal_n(f^{\otimes n})\,,
	\]
	Assume 	
	 that we have constructed $\Zcal$ as a formal power series up to order $N$, i.e. that we have a family of maps
	\[
	\Zcal_{k}:\Vcal^{\otimes k}\rightarrow \Vcal\,,\qquad k\leq N
	\] 
	such that 
	\[
	\widetilde{\Scal}(f)=\Scal\circ \Zcal^N(f)+\Ocal(\lambda^{N+1})\,, 
	\]
	where
	$$\Zcal^N(f)=f+\sum_{k=2}^{N}\frac{\lambda^k}{k!}\Zcal_{k}(f^{\otimes k})$$
	 and $\Zcal^N\in \euR$. Clearly, $\widetilde{\Scal}^N\doteq \Scal\circ \Zcal^N\in \euS$. We can now define
	\be\label{def:ZNplus1}
	\Zcal_{N+1}\doteq \widetilde{\Tcal}_{N+1}-\widetilde{\Tcal}^N_{N+1}\,,
	\ee
	where $\widetilde{\Tcal}_{N+1}$ and $\widetilde{\Tcal}^N_{N+1}$ are coefficients in expansion of $\widetilde{\Scal}$ and $\widetilde{\Scal}^N$ respectively. We need to check that $\Zcal^{N+1}$ defined by
	\[
	\Zcal^{N+1}(f)=\Zcal^{N}(f)+\frac{\lambda^{N+1}}{(N+1)!}\Zcal_{N+1}(f^{\otimes N+1})
	\]
	is an element of $\euR$. The only non-trivial properties to check are {\bf Z\ref{Z:caus:pres}} and {\bf Z\ref{Z:loc}}. First we show that {\bf S\ref{S:loc}} and {\bf S\ref{S:caus:pres}} imply that  {\bf Z\ref{Z:caus:pres}} and {\bf Z\ref{Z:loc}} hold for $f=0$. Note that
	\[
	\widetilde{\Scal}(\lambda(f_1+f_2))=\1+\sum_{n=1}^{\infty}\sum_{k=0}^n \frac{\lambda^n}{k!(n-k)!}\widetilde{\Tcal}_n(f_1^{\otimes k}\otimes f_2^{\otimes n-k})\,,
	\]
	and from {\bf S\ref{S:caus:pres}} follows that if $\supp f_1\dashv \supp f_2$, then 
	\[
	\sum_{k=0}^n \frac{1}{k!(n-k)!}\widetilde{\Tcal}_n(f_1^{\otimes k}\otimes f_2^{\otimes n-k})=\widetilde{\Tcal}_n(f_1^{\otimes n})+\widetilde{\Tcal}_n(f_2^{\otimes n})
	\]
	for all $n\in \NN$, so
	\be\label{Tntilde:caus}
	\sum_{k=1}^{n-1} \frac{1}{k!(n-k)!}\widetilde{\Tcal}_n(f_1^{\otimes k}\otimes f_2^{\otimes n-k})=0\,.
	\ee
	Similarly
	\be\label{TntildeN:caus}
	\sum_{k=1}^{n-1} \frac{1}{k!(n-k)!}\widetilde{\Tcal}^N_n(f_1^{\otimes k}\otimes f_2^{\otimes n-k})=0\,.
	\ee
	Assuming that $\Zcal^N$ satisfies {\bf Z\ref{Z:loc}} with $f=0$, for $\Zcal^{N+1}$ to satisfy {\bf Z\ref{Z:loc}} with $f=0$, we need to prove that
	\[
\Zcal_{N+1}((f_1+f_2)^{\otimes N+1})\\=\Zcal_{N+1}(f_1^{\otimes N+1})+\Zcal_{N+1}(f_2^{\otimes N+1})\,,
	\]
	i.e.
	\[
	\sum_{k=1}^{N}\frac{1}{k!(N+1-k)!}\Zcal_{N+1}(f_1^{\otimes k}\otimes f_2^{\otimes N+1-k})=0
	\]
	This is satisfied, because of
	\eqref{def:ZNplus1} together with \eqref{Tntilde:caus} and \eqref{TntildeN:caus}. 
	
	Now we prove {\bf Z\ref{Z:loc}} for $f=0$. Assume that $f_1\dashv f_2$. We need to show that 
	\[
	\Zcal^{N+1}(f_1)\dashv f_2\,,\quad f_1\dashv \Zcal^{N+1}(f_2)\,.
	\]
	From the definition of $\dashv$ on formal power series and the inductive step assumption, we know that $\Zcal_{n}(f_1^{\otimes n})\dashv f_2$ and $f_1\dashv\Zcal_{n}(f_2^{\otimes n})$
	for all $1\leq n\leq N$. Using \eqref{def:ZNplus1}, all we need to show is that
	\be\label{eq:causpres}
\Zcal_{N+1}(f_1^{\otimes N+1})\dashv f_2\,,\	f_1\dashv \Zcal_{N+1}(f_2^{\otimes N+1})\,,
	\ee
	We use the fact that $\widetilde{\Scal}$ and $\widetilde{\Scal}^N$ are both causality preserving ({\bf S\ref{S:caus:pres}}) so
	\begin{align*}
\widetilde{\Tcal}_{N+1}(f_1^{\otimes N+1})\dashv f_2&\,,\ 
f_1\dashv \widetilde{\Tcal}_{N+1}(f_2^{\otimes N+1})\,,\\ \widetilde{\Tcal}^N_{N+1}(f_1^{\otimes N+1})\dashv f_2&\,,\ 
f_1\dashv \widetilde{\Tcal}^N_{N+1}(f_2^{\otimes N+1})\,,
	\end{align*}
and since $\dashv$ is compatible with addition, the equation \eqref{eq:causpres} follows.

We have thus proven that $Z^{N+1}$ satisfies {\bf Z\ref{Z:caus:pres}} and {\bf Z\ref{Z:loc}} for $f=0$, so from Lemma~\ref{key:lemma} follows that {\bf Z\ref{Z:caus:pres}} and {\bf Z\ref{Z:loc}} hold also for general $f\in \lambda \Vcal[[\lambda]]$.

Finally, we invoke Lemma~\ref{key:lemma2} to conclude that $Z^{N+1}$ preserves $\lambda \Vcal[[\lambda]]$, so is an element of the renormalization group. 
\end{proof}

\section{Local functionals and renormalization}\label{sec:LocFunRen}
\subsection{Smooth local functionals}\label{sec:smooth:func}

As stated in the introduction, local functionals can be characterized by requiring a regularity condition together with the Hammerstein property. For concreteness, we fix a globally hyperbolic spacetime $M$ and we will focus on the example of the real scalar field, i.e. we start with the classical field configuration space $\Ecal\equiv\Ci(M,\RR)$. For future reference, denote $\Dcal\equiv\Ci_c(M,\RR)$, the space of smooth functions on M with compact support.  Classical observables are smooth functionals of $\Ecal$, i.e. elements of $\Ci(\Ecal,\RR)$ (We want the observables to be real-valued, since in the quantum theory the involution is defined as the complex conjugation and we want the observables to be ``self-adjoint'' in the sense that $F^*=F$). Smoothness is understood in the sense of Bastiani \cite{Bas64,Ham,Mil,Neeb06}:
\begin{df}
	Let $\Xcal$ and $\Ycal$ be topological vector spaces, $U \subseteq \Xcal$ an open set and $F:U \rightarrow \Ycal$ a map. The derivative of $f$ at $x\in U$ in the direction\index{derivative!on a locally convex vector space} of $h\in\Xcal$ is defined as
	\be\label{de}
	\left<F^{(1)}(x),h\right> \doteq \lim_{t\rightarrow 0}\frac{1}{t}\left(F(x + th) - F(x)\right)
	\ee
	whenever the limit exists. The function $f$ is called differentiable\index{infinite dimensional!calculus} at $x$ if $\left<F^{(1)}(x),h\right>$ exists for all $h \in \Xcal$. It is called continuously differentiable if it is differentiable at all points of $U$ and
	$F^{(1)}:U\times \Xcal\rightarrow \Ycal, (x,h)\mapsto F^{(1)}(x)(h)$
	is a continuous map. It is called a $\Ccal^1$-map if it is continuous and continuously differentiable. Higher derivatives are defined by
	\[
	\left<F^{(k)}(x),v_1\otimes\dots\otimes v_k\right>\doteq\left. \frac{\partial^k}
	{\partial t_1\dots\partial t_k}F(x+t_1 v_1+\dots + t_k v_k)\right|_{t_1=\dots=t_k=0},
	\]
	and $f$ is $\Ccal^k$ if $f^{(k)}$ is jointly  continuous as a map $U\times \Xcal^k\rightarrow \Ycal$. We say that $f$ is smooth if it is $\Ccal^k$ for all $k\in\NN$. 
\end{df}

The following definition of locality has been proposed in \cite{BDF} and refined in \cite{BDGR}. 

\begin{df}\label{def:LocalFunctional}
	Let $U\subset \Ecal$. A Bastiani smooth functional $F:U\rightarrow \RR$ is \emph{local} if:
	\begin{enumerate}[{\bf LF1}]
		\item $F$ satisfies the (generalized) Hammerstein property (definition~\ref{def:genHamm}) as a map $F:(\Ecal,\tang,+,0)\rightarrow (\CC,+,0)$, where $\tang$ is the disjointness of supports.
			\item For every $\varphi\in U$, the differential 
		$F^{(1)}(\varphi)$ of $F$ at $\varphi$ has an
		empty  wave front set and the map 
		$\varphi\mapsto F^{(1)}(\varphi)$ is Bastiani smooth
		from $U$ to $\Dcal$.\label{technical}
	\end{enumerate}
\end{df}

An alternative definition has been given in \cite{BFR}, where the condition {\bf LF\ref{technical}} is replaced by a weaker condition:
	\begin{enumerate}[{\bf LF1'}]
	\addtocounter{enumi}{1}
	\item Assume that that $F^{(1)}$ is locally bornological (see below) into $\Gamma_c(\wedge^dT^*M\rightarrow M)$, where $d$ is the dimension of $M$.\label{technical2}
\end{enumerate}
\begin{df}
	Let $\mathcal{V}_1$, $\mathcal{V}_2$ be two locally convex vector spaces and $U\subset \mathcal{V}_1$ a non-empty open subset. A map $T:U\rightarrow \mathcal{V}_2$ is said to be locally bornological (into $\mathcal{V}_2$) if for all $x\in U$ there is an open neighborhood $V\subset U$, $x\in V$  such that $T\big|_V$ maps bounded subsets of $V$ into bounded subsets of $\mathcal{V}_2$.
\end{df}
The following results relate the above formulations to the more ``traditional'' notion of locality used in physics.
\begin{thm}[VI.3 in \cite{BDGR}]\label{TheoPrincipal}
	Let $U\subset\Ecal$ and 
	$F:U\to \RR$ be Bastiani smooth. Then, $F$ is local in the sense of definition~\ref{def:LocalFunctional} if and only if
for every $\varphi\in U$, there is a neighborhood 
$V$ of $\varphi$, an integer $k$,
an open subset $\mathcal{V}\subset J^kM$ and a smooth
function $f\in C^\infty(\mathcal{V})$ such that
$x\in M\mapsto f(j^k_x\psi)$ is supported in a compact subset $K\subset M$ and
\begin{align*}
	F(\varphi+\psi) &= F(\varphi)+\int_M f(j^k_x\psi) dx,
\end{align*}
whenever $\varphi+\psi\in V$
and where $j^k_x\psi$ denotes the $k$-jet of $\psi$ at $x$.

Here $J^kM$ denotes the bundle of $k-th$ jets and $j_x^k(\ph)$ is the $k-th$ jet of $\ph$ as a point $x\in M$.
\end{thm}

A generalization of this theorem has been proven in \cite{BFR}.
\begin{thm}[Prop.~2.3.13 in \cite{BFR}]\label{TheoPrincipal2}
	Let $U\subset\Ecal$ and 
	$F:U\to \RR$ be Bastiani smooth. In theorem \ref{TheoPrincipal}, {\bf LF\ref{technical}} can be replaced by a weaker condition {\bf LF\ref{technical2}'}.
\end{thm}	

Clearly, the space of local functionals $\Fcal_{\loc}$ forms a group with addition. In order to equip it with causality or locality structure, we will also need a notion of \textit{spacetime support}: 
\begin{df}
For $F\in\Ci(\Ecal,\RR)$, its support is defined by
\be\label{support}
\supp F\doteq\{ x\in M\mid\forall\ \text{open }U\ni x\ \exists \ph,\psi\in\Ecal,\,
 \supp\psi\subset U,F(\ph+\psi)\not= F(\ph)\}\,.\nonumber
\ee
\end{df}
\subsection{Local $S$-matrices in the functional approach}\label{sec:local:S:matrices}
Using local functionals one can construct a concrete realization of generalized local S-matrices from definition~\ref{Smatrix} that is relevant for building QFT models. In this section we summarize how both classical and quantum dynamics can be nicely described using maps satisfying the Hammerstein property. 
\subsubsection{Classical dynamics}
We introduce generalized Lagrangians, following the approach of \cite{BDF}. 
\begin{df}
	A generalized Lagrangian is a smooth map $L$ from $(\Dcal,\tang,+,0)$ to $(\Fcal_\loc,\tang,+,0)$ (here $\tang$ is disjointness of supports), which is support preserving 
	\[
	\supp(L(f))\subset \supp f\,.
	\]
	(in particular it is a \textit{locality map})
	and satisfies the Hammerstein property, i.e.
	\[
	L(f_1+f+f_2)=L(f_1+f)-L(f)+L(f_2+f)\,,
	\]
	if $f_1\tang f_2$. Let $\euL$ denote the space of generalized Lagrangians.
\end{df}
Following \cite{BF19}, we introduce some notation. 
\begin{df}\label{df:delta:L}
	Let $L\in \euL$, $\ph\in\Ecal$. We define a functional $\delta L(\ph):\Dcal\times\Ecal\rightarrow \RR$ by
	\[
	\delta L(\psi)[\ph]\doteq L(f)[\ph+\psi]-L(f)[\ph]\,, 
	\]
	where $\ph\in\Ecal$, $\psi\in \Dcal$ and $f\equiv 1$ on $\supp \psi$ (the map $\delta L(\psi)[\ph]$ thus defined does not depend on the particular choice of $f$).
\end{df}
The above definition can be turned into a difference quotient and we can use it to introduce the \textit{Euler-Lagrange derivative} of $L$.
\begin{df}
	The \textit{Euler-Lagrange derivative} of $L$ is a map $dL:\Ecal\times\Dcal\rightarrow \RR$ defined by
	\[
	\left<dL(\ph),\psi\right>\doteq \lim_{t\rightarrow 0}
 \tfrac{1}{t}\delta L(t\psi)[\ph]\,,
 \] 
 where $\psi\in\Dcal$, $\ph\in\Ecal$. 
\end{df}
Note that $dL$ can be seen as a 1-form on $\Ecal$ (i.e. as a map from $\Ecal$ to $\Dcal'$). The zero locus of $dL$ is the space of \textit{solutions to the equations of motion}, i.e. $\ph\in\Ecal$ is a solution if
\[
dL(\ph)\equiv 0
\]
as an element of $\Dcal'$.

For the free scalar field the Lagrangian is
\be\label{FreeScalar}
L(f)[\ph]=\frac{1}{2}\int_M\left(\nabla_\nu\ph\nabla^\nu\ph-m^2\ph^2\right)fd\mu_g\,,
\ee
where $\mu_g$ is the invariant measure associated with the metric $g$ of $M$. The equation of motion is  
$$dL(\ph)=P\ph=0\,,$$
 where $P=-(\Box+m^2)$ is (minus) the Klein-Gordon operator. On a globally hyperbolic spacetime $M$, $P$ admits  retarded and advanced Green's functions $\Delta^{\rm R}$, $\Delta^{\rm A}$. These are distinguished by the properties that 
\[
P\circ\Delta^{\rm R/A}=\id_{\Dcal}\,\qquad \Delta^{\rm R/A}\circ(P\big|_{\Dcal})=\id_{\Dcal}\]
and
\begin{align*}
\supp(\Delta^{\rm R})&\subset\{(x,y)\in M^2| y\in J_-(x)\}\,,\\
\supp(\Delta^{\rm A})&\subset\{(x,y)\in M^2| y\in J_+(x)\}\,.
\end{align*}
This is illustrated on the diagram below.
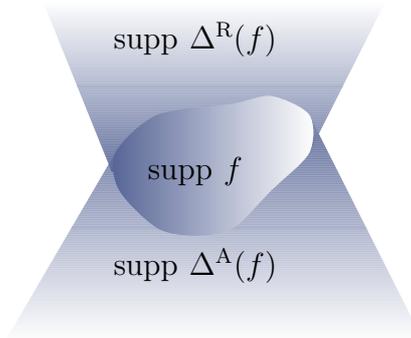
\begin{center}
	\begin{center}
		\begin{tikzpicture}[scale=0.5]
		\shade[shading=axis, top color=blue2!80!white,bottom color=white,opacity=0.8] (-2.1,1) -- (-5,-4) -- (6,-4) -- (3,2);
		\shade[shading=axis, bottom color=blue2!80!white,top color=white,opacity=0.8] (-2.2,0.5) -- (-4,5) -- (5,5) -- (3,1);
		\shade[shading=axis,shading angle=90, left color=blue2!80!white,right color=white] plot[smooth cycle] coordinates{(-1,-1) (-2,0)(-2.1,1)(-1,2)(1,2.3)(2,2.5)(3,2)(3,1)(2,0)(1,-1)(0,-1.2)};
		\draw(0,0.5) node {$\supp\,f$};
		\draw(0,-2) node {$\supp\,\Delta^{\rm A}(f)$};
		\draw(0,4) node {$\supp\,\Delta^{\rm R}(f)$};
		\end{tikzpicture}
	\end{center}
\end{center}

\subsubsection{Quantum dynamics}\label{sec:QD}
Following the approach of \cite{BF19}, we can now define the \textit{dynamical $S$-matrix}. Firstly, note that $G=(\Fcal_\loc,\preceq,+,0)$ is a group with causality, if we define $\preceq$ as the ``not in the future of'' relation on supports of functionals. Moreover, it satisfies {\bf L\ref{L1}}, since for $F_1\preceq F_2$, we can find a Cauchy surface $\Sigma$ separating the supports of $F_1$ and $F_2$ and since $\supp F_1$ and $\supp F_2$ are compact, there exist an open neighborhood of $\Sigma$ that can be foliated by Cauchy surfaces. Using this foliation, we can find a partition of unity $1=\sum_{i=1}^{N} \chi_i$, such that $\chi_i\preceq \supp F_2$ for $i=1,\dots,N-1$, $\supp F_1\preceq \chi_i$ for $i=2,\dots,N$ and for $2\leq i+1<j\leq N$ we have $\chi_i\preceq \chi_j$. Using this partition of unity, we can decompose any local functional $F$ in the way prescribed by {\bf L\ref{L1}}.

Let  $L\in \euL$ (this is interpreted as the Lagrangian of the theory). We then define $\fA_{L}$ using the group algebra over $\CC$ of the free group generated by elements $\Scal(F)$, $F\in G$, modulo relations {\bf S\ref{S:loc}
}, {\bf S\ref{S:start}} and the following relation proposed by \cite{BF19} that encodes the dynamics:
\begin{enumerate}[{\bf S1}]
\addtocounter{enumi}{5}
	\item \label{eq:SD} We have:
\[
\Scal(F)\Scal(\delta L(\ph))=\Scal(F^{\ph}+\delta L(\ph))=\Scal(\delta L(\ph))\Scal(F)\,,
\]
where  $F^{\ph}(\psi)\doteq F(\ph+\psi)$, $\ph,\psi\in \Ecal$ and $\delta L$ is given in Def.~\ref{df:delta:L}.
\end{enumerate}

Physically, this relation can be interpreted as the \textit{Schwinger-Dyson equation} on the level of formal S-matrices.

The construction of the dynamical formal S-matrix presented here allows one to construct interacting nets of algebras satisfying the axioms {\bf HK\ref{HK:iso}-HK\ref{HK:ts}}, starting from the given Lagrangian (see \cite{BF19} for details). This, however, is not sufficient to have a complete physical description of the system, since one still needs to identify physically relevant states. The existence of such states (e.g the vacuum state or thermal states on Minkowski spacetime) for a given theory has not been established yet, but a perturbative procedure is known, starting from states of the free theory.

\subsubsection{Concrete models and states}\label{sec:conc:mod}
In this section we outline another construction of the local $S$-matrix starting from a Lagrangian, which is closely related to that of Section~\ref{sec:QD}. The advantage is that it is more explicit and it comes with an obvious prescription how to define states, perturbatively, for the interacting theory. For the purpose of this review, we will treat only compactly supported interactions. For the discussion of adiabatic limit, see for example \cite{BDF,Vienna}.

We work in the functional formalism. First we define the Pauli-Jordan (commutator) function as the difference of the retarded and advanced Green functions:
\[
\Delta\doteq\Delta^{\rm R}-\Delta^{\rm A}\,.
\]
This gives us the Poisson bracket of $F,G\in\Fcal_{\loc}$ by
\[
\{F,G\}\doteq \left<F^{(1)},\Delta G^{(1)}\right>\,
\]
where $\left<.,.\right>$ denotes the dual pairing between $F^{(1)}(\ph)\in\Dcal$ and $\Delta G^{(1)}(\ph)\in \Ecal$ given by the integration over $M$ with the invariant volume measure $d\mu_g$ induced by the metric. This pairing can be extended also to distributions (as we will often do in the formulas that follow). 

The $\star$-product (deformation of the pointwise product) is defined by
\[
(F\star G)(\ph)\doteq \sum\limits_{n=0}^\infty\frac{\hbar^n
}{n!} \left<F^{(n)}(\ph),W^{\otimes n}G^{(n)}(\ph)\right>\ ,
\]
where $W$ is the 2-point function of some \textit{quasifree Hadamard state} and it differs from $\frac{i}{2}\Delta$ by a symmetric bidistribution:
\[
W=\frac{i}{2}\Delta+H\,.
\]
Without going into technical details, we note that $W$ has the following crucial properties:
\begin{enumerate}[\bf (H 1)]
	\item  $\Delta=2\operatorname{Im} (W)$.
	\item $W$ is a distributional bisolution to the field equation, i.e. $\left<W,Pf\otimes g\right>=0$ and $\left<W,f\otimes Pg\right>=0$ for all $f,g\in\Dcal^{\Co}$.
	\item $W$ is of positive type, meaning that $\left< W,\bar{f}\otimes f\right>\geq 0$, where $\bar{f}$ is the complex conjugate of  $f\in\Dcal^{\Co}$.
	\item \label{WFcond} The wavefront set of $W$ satisfies:
	\[
	\WF(W)=\{(x,k;x',-k')\in \dot{T}^*M^2|\\(x,k)\sim(x',k'), k\in (\overline{V}_+)_x\}\,,
	\]
	where $(\overline{V}_+)_x$ is the closed future lightcone in $T^*_xM$ and	
	the equivalence relation $\sim$ means that there exists a null geodesic strip such that both $(x,k)$ and $(x',k')$ belong to it. Recall that a null geodesic strip is a curve in $T^*M$ of the form $(\gamma(\lambda),k(\lambda))$, $\lambda\in I\subset \RR$, where $\gamma(\lambda)$ is a null geodesic parametrized by $\lambda$ and $k(\lambda)$ is given by $k(\la)=g(\dot{\gamma}(\la),\cdot)$. See \cite{Rad} for a details.
\end{enumerate} 
\begin{df}\label{df:Microcausal}
	A functional $F\in\Ci(\Ecal,\RR)$ is called \emph{microcausal} ($F\in \Fcal_{\mc}$) if it is compactly supported and satisfies
	\be\label{mlsc}
	\WF(F^{(n)}(\ph))\subset \Xi_n,\quad\forall n\in\NN,\ \forall\ph\in\Ecal\,,
	\ee
	where $\Xi_n$ is an open cone defined as 
	\be\label{cone}
	\Xi_n\doteq T^*M^n\setminus\{(x_1,\dots,x_n;k_1,\dots,k_n)\,|\, (k_1,\dots,k_n)\in (\overline{V}_+^n \cup \overline{V}_-^n)_{(x_1,\dots,x_n)}\}\,,
	\ee
	where $(\overline{V}_{\pm})_x$ is the closed future/past lightcone.
\end{df}
It follows from {\bf H\ref{WFcond}} and the definition~\ref{df:Microcausal} of microcausal functionals, that the star product $\star$ can be extended to $\Fcal_{\mc}[[\hbar]]$. 

To illustrate concepts presented in section \ref{prg} on a concrete example, consider the space of Laurent series in the formal parameter $\hbar$ with coeffcients in formal power series in $\lambda$, i.e.
$\fA=(\Fcal_\mc[[\lambda]]((\hbar)),\star_H)$. The relevant causality relation is defined as follows: 
\[
F\dashv G\  \textrm{iff}\ \supp(F)\preceq \supp(G)\,.
\]
The relevant locality relation $\tang$ is the disjointness of supports of functionals.

Set $\Vcal\equiv\Fcal_\loc$ and for $V\in\Fcal_{\loc}$, set
\[
\Scal(\lambda V)\doteq e_\TT^{i\lambda V/\hbar}=\1+\sum_{n=1}^{\infty} \left(\frac{i\lambda}{\hbar}\right)^n\frac{1}{n!}\TT_n(V^{\otimes n})
\,,
\]
as the formal $S$-matrix. For $\Scal$ given in terms of $\TT_{n}$s to satisfy {\bf S\ref{S:loc}}, we need to require that:
\begin{enumerate}[{\bf T1}]
	\item \label{factorisation:T} We have:
	\[
	\TT_n(V_1,\dots, V_n)=\TT_k(V_1,\dots, V_k)\star \TT_k(V_{k+1},\dots, V_n)\,,
	\]
	if $\supp V_{k+1},\dots,\supp V_n$ are not later than $\supp V_{1},\dots,\supp V_k$. 
\end{enumerate}
This is the key property (called \textit{causal factorisation}) for setting up the inductive procedure of Epstein and Glaser (see \cite{EG} and a comprehensive review \cite{Due19}) that allows one to prove the existence of formal S-matrices in the sense of definition~\ref{Smatrix}.

To construct an explicit example of a local S-matrix, some further structures are needed. Given $W$, one defines the corresponding Feynman propagator by
\[
\Delta^{\rm F}=\frac{i}{2}(\Delta^{\rm A}+\Delta^{\rm R})+H
\]
Let $F_1,\dots,F_n\in\Fcal_{\loc}^{\tang k}$, where $\tang$ is disjointness of supports and we use the notation of definition~\ref{df:localityset}. The $n$-fold time-ordered product of $F_1,\dots,F_n$ is defined by
\[
\TT_n(F_1,\dots,F_n)=m_n\circ e^{\sum_{1\leq i<j\leq n}D^{\rm F}_{ij}}(F_1\otimes\dots\otimes F_n)\,,
\]
where 
\[
D^{\rm F}_{ij}\doteq \left\langle\Delta^{\mathrm{F}},\frac{\delta}{\delta \ph_i}\otimes\frac{\delta}{\delta\ph_j}\right\rangle
\]
and $m_n$ is the 	a pullback through the diagonal map $\Ecal\rightarrow\Ecal^{\otimes n}$, $\ph\mapsto \ph\otimes\dots\otimes \ph$. 

	The family of $n$-fold time-ordered products defined above can be used to construct a locality algebra  $(\Tcal(\Fcal^{\rm pds}_{\ml}[[\hbar]]),\tang,\T,+,0,1)$, where:
	\begin{itemize}
		\item $\Fcal^{\rm pds}_{\ml}$ is the space of functionals that arise as finite sums of local functionals and products of local functionals with $F(0)=0$ and pairwise disjoint supports.
		\item The map  $\Tcal$ is defined by
	\be\label{eq:Tmap}
	\Tcal=\bigoplus_{n=0}^\infty \TT_n\circ \beta\,,
	\ee
	where $\beta$ is the inverse of multiplication on multilocal functionals (as defined in \cite{FR3})
	\item The product is:
	\[
	F\T G= \TT(\TT^{-1}F\cdot \TT^{-1} G)= \sum\limits_{n=0}^\infty\frac{\hbar^n
	}{n!} \left<F^{(n)},(\Delta^{\rm F})^{\otimes n}G^{(n)}\right>\,.
	\]
	Note that, as in \cite{CGPZ}, the algebra product is defined only for $\tang$-related elements. 
	\end{itemize} 
To construct $\Scal:\lambda \Fcal_{\loc}[[\lambda]]\rightarrow \Fcal_{\mc}[[\lambda]]((\hbar))$, one needs to extend the maps $\Tcal_n$ to arbitrary local functionals (i.e. the full $\Fcal_{\loc}^{n}$). This has been done, under some additional assumptions on  $\Tcal_n$s in \cite{HW}. There, it has been shown that maps $\Tcal_n:\Fcal_{\loc}^{n}\rightarrow \fA$ satisfying {\bf T\ref{factorisation:T}} and some further, physically motivated, properties exist, but are not unique. The ambiguity of extensions of $\TT_n$s is governed by the renormalization group $\euR$ (as two $S$-matrices have to differ by $\Zcal\in\euR$ by virtue of Theorem~\ref{thm:main}).

Given such extended (i.e.renormalized) $\TT_ns$, one constructs the algebra $$(\TT(\Fcal_{\ml}[[\hbar]]),\T,+,0,1)\,,$$ where $\Fcal_{\ml}$ is analogous to $\Fcal^{\rm pds}_{\ml}$, but without the supports restriction and $\TT$ is the same as in \eqref{eq:Tmap} (see \cite{FR3}).

States on $\fA$ can be defined using evaluation functionals. For example, the evaluation at $\ph=0$ corresponds to the expectation value in the quasifree Hadamard state, whose 2-point function $W$ has been used to define $\star$.

In the next step one introduces the interacting fields. Given interaction $V\in\Fcal_\loc$ and $F\in \Fcal_{\loc}$, we introduce the \textit{relative S-matrix} using Def.~\ref{def:relative}, i.e.
\[
\Scal_{\lambda V}(\mu F)\doteq \Scal(\lambda V)^{-1}\star\Scal(\lambda V +\mu F)\,,
\]
understood as a formal power series in both $\lambda$ and $\mu$ and a Laurent series in $\hbar$.

The formal $S$-matrix is a generating functional for the interacting fields. We define the interacting observable corresponding to $F$ by
\[
F_{\rm int}\doteq -i\hbar \frac{d}{d\mu} \Scal_{\lambda V}(\mu F)\big|_{\mu =0}\,.
\]
Correlation functions of interacting observables corresponding to $F_1,\dots,F_n$ can then be easily computed by means of 
\[
\omega_{\rm int}(F_1,\dots,F_n)=
({F_1}_{\rm int}\star\dots\star {F_n}_{\rm int})\big|_{\ph =0}\,.
\]
\section{Conclusions}
In this paper we have shown how the intuitive notion of ``locality'' in classical and quantum field theory is nicely captured by the appropriate generalization of the Hammerstein property. We hope that further investigation of structures with that property will lead to a better understanding of locality in physics.
\section*{Acknowledgments}
I would like to thank Christian Brouder, Pierre Clavier and Sylvie Paycha for fascinating discussions about locality that we had in Potsdam and which prompted me to think about this structure again. I would also like to thank Camille Laurent-Gengoux for discussions and hospitality in Paris and finally Marco Benini, Michael D\"utsch, Klaus Fredenhagen and Alexander Schenkel for some very useful comments. Part of this research has been conducted while visiting the Perimeter Institute for Theoretical Physics (Waterloo), it was also supported by the EPSRC grant \verb|EP/P021204/1|.

\bibliographystyle{amsplain}
\bibliography{References}

\end{document}